\documentclass[a4paper,10pt]{article}
\usepackage{amsmath}
 \usepackage{amssymb,amsthm}
 \usepackage{algorithm}
 \usepackage{algorithmic}
\usepackage{comment}
\usepackage{tikz}
\usepackage{appendix}
\usepackage{multirow}
\usetikzlibrary{shapes.arrows}

 \newtheorem{theorem}{Theorem}[section]
  \newtheorem{lemma}[theorem]{Lemma}
  \newtheorem{claim}[theorem]{Claim}
\DeclareMathOperator{\alg}{ALG}
\DeclareMathOperator{\opt}{OPT}
\DeclareMathOperator{\adv}{ADV}
\title{On Randomized Algorithms for Matching in the Online Preemptive Model}
 \author{Ashish Chiplunkar\\
 \{ashish.chiplunkar@gmail.com\}
 \and Sumedh Tirodkar\\
 \{sumedht@cse.iitb.ac.in\}
 \and Sundar Vishwanathan\\
 \{sundar@cse.iitb.ac.in\}}
 \date{}
\begin{document}

\maketitle

\begin{abstract}
We investigate the power of randomized algorithms for the maximum cardinality matching (MCM) and the maximum weight matching (MWM) problems in the online preemptive model. In this model, the edges of a graph are revealed one by one and the algorithm is required to always maintain a valid matching. On seeing an edge, the algorithm has to either accept or reject the edge. If accepted, then the adjacent edges are discarded. The complexity of the problem is settled for deterministic algorithms~\cite{mcgregor,varadaraja}.

Almost nothing is known for randomized algorithms. A lower bound of $1.693$ is known for MCM with a trivial upper bound of two. An upper bound of $5.356$ is known for MWM. We initiate a systematic study of the same in this paper with an aim to isolate and understand the difficulty.
%We present a primal-dual analysis for the deterministic algorithm due to~\cite{mcgregor}, and extend this technique to barely random algorithms for MCM on paths and growing trees. Next, we identify certain natural classes of randomized online algorithms, and of input models, and prove lower bounds on the competitive ratio achievable for these classes. We also present the best possible $\frac{4}{3}$-competitive randomized algorithm for MCM on paths.
We begin with  a primal-dual analysis of the deterministic algorithm due to~\cite{mcgregor}. All deterministic lower bounds are on instances which are trees at 
every step. For this class of (unweighted) graphs we present a randomized algorithm which is $\frac{28}{15}$-competitive. The analysis is a considerable extension of the (simple) primal-dual analysis for the deterministic case. The key new technique is that the distribution of primal charge to  dual variables depends on the ``neighborhood'' and needs to be done after having seen the  entire input. The assignment is asymmetric: in that edges may assign different charges to the two end-points. Also the proof depends on a non-trivial structural statement on the performance of the algorithm on the input tree.

The other main result of this paper is an extension of  the deterministic lower bound of Varadaraja~\cite{varadaraja} to a natural  class of  randomized algorithms which decide whether to accept a new edge  or not using {\em independent} random choices. This indicates that randomized algorithms will have to use {\em dependent} coin tosses to succeed. Indeed, the few known randomized algorithms, even in very restricted models follow this.

We also present the best possible $\frac{4}{3}$-competitive randomized algorithm for MCM on paths.
\end{abstract}

\section{Introduction}
Matching has been a central problem in combinatorial optimization. Indeed, algorithm design in various models of computations, sequential, parallel, streaming, etc., have been influenced by techniques used for matching. We study the maximum cardinality matching (MCM) and the maximum weight matching (MWM) problems in the online preemptive model. In this model, edges $e_1,\dots,$  $e_m$ of a graph, possibly weighted, are presented one by one. An algorithm is required to output a matching $M_i$ after the arrival of each edge $e_i$. This model constrains an algorithm to accept/reject an edge as soon as it is revealed. If accepted, the adjacent edges, if any, have to be discarded from $M_i$.

 An algorithm is said to have a \textit{competitive ratio} $\alpha$ if 
 the cost of the matching maintained by the algorithm is at
 least $\frac{1}{\alpha}$ times the cost of the offline optimum over all
 inputs.
The deterministic complexity of this problem is settled. 
For maximum cardinality matching (MCM), it is an easy exercise to prove
a tight bound of two.

The weighted version (MWM) is more difficult. 
Improving an earlier result of Feigenbaum et al, McGregor~\cite{mcgregor} gave a deterministic algorithm together with an ingenious analysis
to get a competitive ratio of $3+2\sqrt{2}\approx5.828$.
Later,  this was proved to be  optimal by Varadaraja~\cite{varadaraja}. 

Very little is known on the power of randomness for this problem.
Recently, Epstein et al.~\cite{epstein} proved a lower bound of 
$1+\ln 2 \approx1.693$ on the competitive ratio of randomized algorithms
for MCM. This is the best lower bound known even for MWM.
Epstein et al.~\cite{epstein} also give a $5.356$-competitive randomized
algorithm for MWM. 

In this paper, we initiate a systematic study of the power of randomness 
for this problem.  Our main contribution is perhaps to throw some light
on where lies the difficulty. We first give an analysis of McGregor's
algorithm using the traditional Primal-Dual framework (see Appendix~\ref{pd}).  All lower bounds
for deterministic algorithms (both for MCM and MWM) employ
{\em growing trees}.
That is,  the input graph is a tree at every stage. It is then natural to start our investigation for this class of inputs. For this class, we give a 
randomized algorithm (that uses two bits of randomness) that is 
$\frac{28}{15}$ competitive. While this result is modest, already the
analysis is considerably more involved than the
traditional primal dual analysis. In the traditional primal dual analysis
of the matching problem, the primal charge (every selected edge 
contributes one to the charge) is distributed (perhaps equally) to
the two end-points. In the online case, this is usually done as
the algorithm proceeds. Our assignment depends on the structure
of the final tree, so this assignment happens at the end.
Our charge distribution is {\em not} symmetric. It depends on
the position of the edge in the tree (we make this clear in 
the analysis) as also the behavior of neighboring edges. The
main technical lemma shows that the charge distribution will
depend on a neighborhood of distance at most four.
We also note that these algorithms are (restricted versions of)
randomized greedy algorithms even in the offline setting. Obtaining an approximation ratio less than two for general graphs, even in the offline setting is
a notorious problem. See \cite{poloczek,chan} for a glimpse of the difficulty.

The optimal maximal matching algorithm for MCM, and McGregor's~\cite{mcgregor} optimal deterministic algorithm for MWM are both local algorithms. The choice of whether a new edge should accepted or rejected is  based only on the weight of the new edge and the weight of the conflicting edges, if any, in the current matching.

It is natural to add randomness to such local algorithms, and to ask 
whether they do better than the known deterministic lower bounds. An
obvious way to add randomness is to accept/reject the new edge with 
certain probability, which is only dependent on the new edge and 
the conflicting edges in the current matching.
The choice of adding a new edge is independent of the previous 
coin tosses used by the algorithm. We call such algorithms 
{\em randomized local algorithms}. We show that randomized local 
algorithms cannot do better than optimal deterministic algorithms. 
This indicates that randomized algorithms may have to use dependent coin
tosses to get better approximation ratios. 
Indeed, the algorithm by Epstein et al. does this. So does
our randomized algorithms.

The randomized algorithm of Epstein et al.~\cite{epstein} works as 
follows. For a parameter $\theta$, they round the weights of the edges 
to powers of $\theta$ randomly, and then they update the matching using 
a deterministic algorithm. The weights get distorted by a factor $\frac{\theta\ln\theta}{\theta-1}$ in the rounding step, and the deterministic algorithm has a competitive ratio of $2+\frac{2}{\theta-2}$ on \textit{$\theta$-structured graphs}, i.e., graphs with edge weights being powers of $\theta$. The overall competitive ratio of the randomized algorithm is $\frac{\theta \ln \theta}{\theta-1}\cdot \left(2+\frac{2}{\theta-2}\right)$ which is minimized at $\theta\approx5.356$. 
A natural approach  to reducing this competitive ratio is to improve the approximation ratio for $\theta$ structured graphs. However, we prove that the competitive ratio $2+\frac{2}{\theta-2}$ is tight for $\theta$-structured graphs, as long as $\theta\geq 4$, for deterministic algorithms.

One (minor) contribution of this paper is
a randomized algorithms for MCM on paths, 
that achieves a competitive ratio of $\frac{4}{3}$, 
with a matching lower bound. 

The other (minor) contribution of this paper  is to highlight
model specific bounds.  There is a  difference in the models
in which the lower and upper bounds have been proved
and this may be one reason for the large gaps. 
%A primal-dual based analysis of McGregor's deterministic algorithm~\cite{mcgregor} is presented in the Appendix (section~\ref{pd}).
\begin{comment}
\subsection{Organization of the paper}
Due to space restrictions, the two main results, the
upper bound on growing trees and the lower bound for
randomized algorithms which use independent coin tosses
are presented in the main body of the paper. An analysis
of McGregor's algorithm, 
an optimal randomized algorithm for MCM on paths,
and on bounded degree trees is moved to the appendix.
In section~\ref{lb}, we present various lower bounds on specific type of algorithms and inputs. 
\end{comment}
\section{Barely Random Algorithms for MCM}
In this section, we present barely random algorithms, that is,
algorithms that use a constant number of random bits, for MCM on growing trees. 

The ideal way to read the paper, for a reader of leisure, is to first read our analysis of McGregor's algorithm (presented in Appendix~\ref{pd}), then the analysis of the algorithm for trees with maximum vertex degree three (presented in Appendix~\ref{ub2}) and then this section. The dual variable management which is the key contribution gets progressively more complicated. It is local in the first two cases. The Appendix~\ref{sa_gt} also gives an example which shows why a non-local analysis is needed. Here are the well known Primal and Dual formulations of the matching problem. The primal formulation is known to be optimum for bipartite graphs. For general graphs, odd set constraints have to be added. But they are not needed in this paper.
 \begin{center}
 \begin{tabular}{c|c}
 Primal LP & Dual LP \\\hline
 $\max \sum_e x_e$ & $\min \sum_v y_v$\\
 $\forall v:\sum_{v\in e}x_e \leq 1$ & $\forall e: y_u + y_v \geq 1$\\
 $x_e\geq 0$ & $y_v \geq 0$
 \end{tabular}
\end{center}
\subsection{Randomized Algorithm for MCM on Growing Trees}\label{ub3}
 In this section, by using only two bits of randomness, we beat the deterministic lower bound of $2$ for MCM on growing trees. 
   \begin{algorithm}[H]
  \caption{Randomized Algorithm for Growing Trees}
   \begin{enumerate}
    \item The algorithm maintains four matchings: $M_1,M_2,M_3,$ and $M_4$.
    \item On receipt of an edge $e$, the processing happens in two phases.
   \begin{enumerate}
	   \item {\bf The augment phase.} The new edge $e$
		   is added to each $M_i$ in which there are no 
		   edges  adjacent to $e$.
	   \item {\bf The switching phase.}
		   For $i=2,3,4$, in order, $e$ is added to
		   $M_i$ (if it was not added in the previous 
		   phase) and the conflicting edge is discarded,
		   provided it decreases the quantity 
		   $\sum_{i,j\in[4],i\neq j}|M_i\cap M_j|$.
   \end{enumerate}
   \item Output  matching $M_i$ with probability $\frac{1}{4}$.
   \end{enumerate}
 \end{algorithm}
 We begin by assuming (we justify this below) that all edges that do not belong to any matching are leaf edges. This helps in simplifying the analysis. Suppose that there is an edge $e$ which does not belong to any matching, but is not a leaf edge. By removing $e$, the tree is partitioned into two subtrees. The edge $e$ is added to the tree in which it has $4$ neighboring edges. (There must be such a subtree, see next para.)
 Each tree is analysed separately.
 
 We will say that a vertex(/an edge) is {\em covered } by a matching $M_i$ if there is an edge in $M_i$ which is incident on(/adjacent to) the vertex(/edge). We also say that an edge is {\em covered } by a matching $M_i$ if it belongs to $M_i$.
We begin with the  following observations.
\begin{itemize}
 \item After an edge is revealed, its end points are
	 covered by all $4$ matchings. 
 \item An edge $e$ that does not belong to any matching has $4$ edges 
	 incident on one of its end points such that each of these edges
	 belong to a distinct matching.
	 This  holds when the edge is revealed, and does not change subsequently.
%  \item The vertices of an internal edge are covered by $4$ matchings, i.e. either there are $4$ distinct edges incident on these vertices and each of them belongs to some matching, or there are less than $4$ edges incident on these vertices, but some edge/edges belongs/belong to multiple matchings.
%  \item Every edge in the tree has to be covered by at least $3$ matchings. An edge covered by only $3$ matchings is defined to be a ``bad'' edge. A ``bad'' edge can only be incident on a vertex of degree $2$ or $3$.
%  \item An edge belonging to $3$ matchings can only be a leaf edge.
\end{itemize}
An edge is called {\em internal} if there are edges incident on both its end points. An edge is called {\em bad} if its end points are covered by only $3$ matchings. 

We begin by proving some properties about the algorithm. The key structural lemma that keeps ``influences'' of bad edges local is given below.  The two assertions in the Lemma have to be proved together by induction.
\begin{lemma}\label{internal}
\begin{enumerate}
 \item An internal edge is covered by at least four matchings (when counted with multiplicities). It is not necessary that these four edges be in distinct matchings.
 \item If $p,q$ and $r$ are three consecutive vertices on a path, then bad edges cannot be incident on all $3$ of these vertices, (as in figure~\ref{IC}).
\end{enumerate}
\end{lemma}
The proof of this lemma is in the Appendix~\ref{pf_inbad}.
\begin{figure}
\centering
\begin{tikzpicture}
\fill [color=black] (0,0) circle (1pt);
\fill [color=black] (2,0) circle (1pt);
\fill [color=black] (4,0) circle (1pt);
\fill [color=black] (6,0) circle (1pt);
\fill [color=black] (8,0) circle (1pt);
\fill [color=black] (2,2) circle (1pt);
\fill [color=black] (4,2) circle (1pt);
\fill [color=black] (6,2) circle (1pt);

\draw (0,0) -- (2,0);
\draw (2,0) -- (4,0);
\draw (4,0) -- (6,0);
\draw (6,0) -- (8,0);

\draw (2,0) -- (2,2);
\draw (4,0) -- (4,2);
\draw (6,0) -- (6,2);

\draw (1.8,0) node[anchor=north west] {$p$};
\draw (3.8,0) node[anchor=north west] {$q$};
\draw (5.8,0) node[anchor=north west] {$r$};

\draw (0.9,1.3) node[anchor=north west] {``bad''};
\draw (2.9,1.3) node[anchor=north west] {``bad''};
\draw (4.9,1.3) node[anchor=north west] {``bad''};

\end{tikzpicture}
\caption{Forbidden Configuration}\label{IC}
\end{figure}
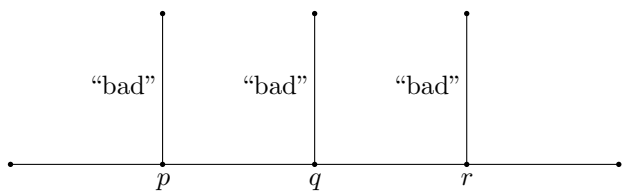

\begin{theorem}
 The randomized algorithm for finding MCM on growing trees is $\frac{28}{15}$-competitive.
\end{theorem}
A local analysis like the one in Appendix~\ref{ub2} will not work here. For a reason, see Appendix~\ref{sa_gt}. The analysis of this algorithm proceeds in two steps. Once all edges have been seen, we impose a partial order on the vertices of the tree and then with the help of this partial order, we distribute the  primal charge to the dual variables, and use the primal-dual framework to infer the competitive ratio. If every edge had four adjacent edges in some matching (counted with multiplicities) then the distribution of dual charge is easy. However we do have edges which have only three adjacent edges in matchings. We would like the
edges in matchings to contribute more to the end-points of these edges. Then, the charge on the other end-point would be less and we need to balance this through other edges. Details follow.\\
\textbf{Ranks:} Consider a vertex $v$. Let $v_1,\dots,v_k$ be the neighbors of $v$. For each $i$, let $d_i$ denote the maximum distance from $v$ to any leaf if there was no edge between $v$ and $v_i$.The rank of $v$ is defined as the minimum of all the $d_i$. Observe that the rank of $v$ is one plus the second highest rank among the neighbors of $v$. Thus there can be at most one neighbor of vertex $v$ which has rank at least the rank of $v$. All leaves have rank $0$. Rank $1$ vertices have at most one non-leaf neighbor.
\begin{lemma}\label{tree1}
There exists an assignment of the primal charge amongst the dual variables such that the dual constraint for each edge $e\equiv (u,v)$ is satisfied at least $\frac{15}{28}$ in expectation, i.e. $\mathbb{E}[y_u+y_v]\geq \frac{15}{28}$.
 \end{lemma}
 
\begin{proof}
  Consider an edge $e\equiv (u,v)$ where rank of $u$ is $i$ and rank of $v$ is $j$. We will show that $y_u+y_v\geq 2+\epsilon$ for such an edge, when summed over all four matchings. The value of $\epsilon$ is chosen later. The proof is by induction on the lexicographic order of $<j, i>$, $j\geq i$. \\ 
\textbf{Dual Variable Management:} Consider an edge $e$ from a vertex of rank $i$ to a vertex of rank $j$, such that $i\leq j$. This edge will distribute its primal weight between its end-points. The exact values are discussed in the proof of the claim below.
	In general, we look to transfer all of the primal 
	charge to the higher ranked vertex. But this does
	not work and we need a finer strategy. This is 
	detailed below.
\begin{itemize}
 \item If $e$ does not belong to any matching, then it does not contribute to the value of dual variables.
 \item If $e$ belongs to a single matching then, depending
	 on the situation, one of $0$, $\epsilon$ or $2\epsilon$
	 of its primal charge will be assigned to the rank $i$
	 vertex and rest will be assigned to the rank $j$ vertex.
	 The small constant $\epsilon$ is determined later.
 \item If $e$ belongs to two matchings, then at 
	 most $3\epsilon$ of its primal charge will be
	 assigned to the rank $i$ vertex as required. The
	 rest is assigned to the rank $j$ vertex.
 \item If $e$ belongs to three or four matchings, then its entire primal 
	 charge is assigned to the rank $j$ vertex.
 \end{itemize}

  The analysis breaks up into six cases.
  
  \textbf{Case 1.} Suppose $e$ does not belong to any matching. Then it must be
	 a leaf edge. Hence, $i=0$. There must be $4$ edges incident on 
	 $v$ besides $e$, each belonging to a distinct matching.
	 Of these $4$, at least $3$ say  $e_1$, $e_2$, and $e_3$, must be
	 from lower ranked vertices to the rank $j$ vertex $v$.
	 The edges $e_1$, $e_2$, and $e_3$, each assign a charge
	 of $1-2\epsilon$
	 to $y_v$. Therefore, $y_u+y_v\geq 3-6\epsilon \geq 2+\epsilon$.
	 
   \textbf{Case 2.} Suppose $e$ is a bad edge that belongs to a single matching. Since no internal edge can be a bad edge, $i=0$. This implies (Lemma~\ref{internal}) that, there is an edge  $e_1$ from a rank $j-1$ vertex to $v$, which belongs to a single	  matching. Also, there is an edge $e_2$, from $v$ to a higher ranked vertex, which also belongs to a single matching. The edge $e$ assigns a charge of $1$ to $y_v$. If $e_1$ assigns a charge of $1$ (or $1-\epsilon$) to $y_v$, then $e_2$ assigns $\epsilon$ (or $2\epsilon$ respectively) to $y_v$.  In either case, $y_u+y_v=2+\epsilon$.
	  %   The primal charge of $e$, and the remaining primal charge of $e_1$ is transferred to $y_v$. Suppose, $e_1$ transfers a charge $\epsilon$ to a lower ranked vertex. This happens if there was a ``bad'' edge incident on the rank $j-1$ vertex of $e_1$. Then $2\epsilon$ is borrowed from $e_2$ and given to $y_v$. Else, $\epsilon$ is borrowed from $e_2$ and given to $y_v$. Thus, $y_u+y_v=2+\epsilon$. 
  The key fact is  that $e_1$ could not have assigned $2\epsilon$ to a
  lower ranked vertex. Since, then, by Lemma~\ref{internal}, $e$ cannot
  be a bad edge.
  
  \textbf{Case 3.} Suppose $e$ is not a bad edge, and it belongs to a single matching. \\
   \textit{Case 3(a).} $i=0$. There are two sub cases. 
   \begin{itemize}
    \item There is an edge $e_1$ from some rank $j-1$ vertex to $v$
		which belongs to $2$ matchings, or there are two other edges
		$e_2$ and $e_3$ from some lower ranked vertices to $v$,
		each belonging to separate matchings. The edge $e$ assigns
		a charge of $1$ to $y_v$. Either $e_1$ assigns a charge of
		at least $2-3\epsilon$ to $y_v$, or $e_2$ and $e_3$ assign a
		charge of at least $1-2\epsilon$ each, to $y_v$.
		In either case, $y_u+y_v\geq 3-4\epsilon\geq 2+\epsilon$.
    
    \item There is one edge $e_1$, from a rank $j-1$ vertex to $v$,
		which belongs to a single matching, and there is one edge $e_2$,
		from $v$ to a higher ranked vertex, which belongs to $2$
		matchings. 
		The edge $e$ assigns a charge of $1$ to $y_v$.
		If $e_1$ assigns a charge of $1$
		(or $1-\epsilon$ or $1-2\epsilon$) to $y_v$, then $e_2$ 
		assigns $\epsilon$ (or $2\epsilon$ or $3\epsilon$ respectively) to $y_v$. In either case, $y_u+y_v=2+\epsilon$.
\end{itemize}
   \textit{Case 3(b).} $i>0$. There are two sub cases.
    \begin{itemize}
      \item There are at least two edges $e_1$ and $e_2$ from lower ranked
		  vertices to $u$, and one edge $e_3$ from $v$ to a higher ranked
		  vertex. Each of these edges are in one matching only 
		  (not necessarily the same matching).
      \item There is one edge $e_4$ from a vertex of lower rank to $u$, at least one edge $e_5$ from a lower ranked vertex to $v$, and one edge $e_6$ from $v$ to a vertex of higher rank. All these edges belong to a single matching (not necessarily the same).
    \end{itemize}
    The edge $e$ assigns a charge of $1$ among $y_u$ and $y_v$. 
	If $e_1$ and $e_2$ assign a charge of at least $1-2\epsilon$ each,
	to $y_u$, then $y_u+y_v\geq3-4\epsilon\geq 2+\epsilon$. 
	Similarly, if $e_4$ assigns a charge of at least $1-2\epsilon$
	to $y_u$, and $e_5$ assigns a charge of at least $1-2\epsilon$ 
	to $y_v$, then $y_u+y_v\geq3-4\epsilon\geq 2+\epsilon$.
    
  \textbf{Case 4.} Suppose $e$ is a bad edge that belongs to two matchings.
	  Then $i=0$. This implies that there is an edge $e_1$, from $v$ to a 
	  vertex of higher rank which belongs to a single matching.
	  The edge $e$ assigns a charge of $2$ to $y_v$, and the edge $e_1$ 
	  assigns a charge of $\epsilon$ to $y_v$. Thus, $y_u+y_v=2+\epsilon$. 
  
  \textbf{Case 5.} Suppose $e$ is not a bad edge and it belongs to two matchings.
	  This means that  either there is an edge $e_1$ from a lower ranked
	  vertex to $u$, which belongs to at least one matching, 
	  or there is an edge from some lower ranked vertex to $v$
	  that belongs to at least one matching,
	  or there is an edge from $v$ to some higher ranked vertex which
	  belongs to two matchings. The edge $e$ assigns a charge of 
	  $2$ among $y_u$ and $y_v$. The neighboring edges assign
	  a charge of $\epsilon$ to $y_u$ or $y_v$ (depending on which vertex it is incident), to give  $y_u+y_v\geq 2+\epsilon$.
	  
  \textbf{Case 6.} Suppose, $e$ belongs to $3$ or $4$ matchings, then trivially $y_u+y_v\geq 2+\epsilon$.
  From the above conditions, the best value for the competitive ratio is obtained when $\epsilon=\frac{1}{7}$, yielding $\mathbb{E}[y_u+y_v]\geq \frac{15}{28}$.
\end{proof}
Lemma~\ref{tree1} implies that the competitive ratio of the algorithm is at most $\frac{28}{15}$.
\section{Lower Bounds}\label{lb}

\subsection{Lower Bound for MWM}\label{sec_lb_mwm}

In this section, we prove a lower bound on the competitive ratio of a natural class of randomized algorithms in the online preemptive model for MWM. The algorithms in this class, which we call \textit{local} algorithms, have the property that their decision to accept or to reject a new edge is completely determined by the weights of the new edge and the conflicting edges in the matching maintained by the algorithm. Indeed, the optimal deterministic algorithm by McGregor \cite{mcgregor} is a local algorithm. The notion of locality can be extended to randomized algorithms as well. In case of {\em randomized local algorithms}, the event that a new edge is accepted is independent of all such previous events, given the current matching maintained by the algorithm. Furthermore, the probability of this event is completely determined by the weight of the new edge and the conflicting edges in the matching maintained by the algorithm. Given that the optimal $(3+2\sqrt{2})$-competitive deterministic algorithm for MWM is a local algorithm, it is natural to ask whether randomized local algorithms can beat the deterministic lower bound of $(3+2\sqrt{2})$ by Varadaraja \cite{varadaraja}. We answer this question in the negative, and prove the following theorem.

\begin{theorem}\label{thm_local}
No randomized local algorithm for the MWM problem can have a competitive ratio less than $\alpha=3+2\sqrt{2}\approx5.828$.
\end{theorem}

Note that the randomized algorithm by Epstein et al. \cite{epstein} does not fall in this category, since the decision of accepting or rejecting a new edge is also dependent on the outcome of the coins tossed at the beginning of the run of the algorithm. (For details, see Section 3 of \cite{epstein}.) In order to prove Theorem \ref{thm_local}, we will crucially use the following lemma, which is a consequence of Section 4 of \cite{varadaraja}.

\begin{lemma}\label{lem_varadaraja}
If there exists an infinite sequence $(x_n)_{n\in\mathbb{N}}$ of positive real numbers such that for all $n$, $\beta x_n\geq\sum_{i=1}^{n+1}x_i+x_{n+1}$, then $\beta\geq3+2\sqrt{2}$.
\end{lemma}

\subsubsection{Characterization of local randomized algorithms}

Suppose, for a contradiction, that there exists a randomized local algorithm $\mathcal{A}$ with a competitive ratio $\beta<\alpha=3+2\sqrt{2}$, $\beta\geq1$. Define the constant $\gamma$ to be
\[\gamma=\frac{\beta\left(1-\frac{1}{\alpha}\right)}{\left(1-\frac{\beta}{\alpha}\right)}=\frac{\beta(\alpha-1)}{\alpha-\beta}\geq1>\frac{1}{\alpha}\]
% The behavior of $\mathcal{A}$ is completely described using three functions: $f_0:\mathbb{R}^+\rightarrow[0,1]$, $f_1: \mathbb{R}^+ \times \mathbb{R}^+ \rightarrow [0,1]$ and $f_2: \mathbb{R}^+ \times \mathbb{R}^+ \times \mathbb{R}^+ \rightarrow [0,1]$, which we describe below. 
For $i=0,1,2$, if $w$ is the weight of a new edge and it has $i$ conflicting edges, in the current matching, of weights $w_1,\ldots,w_i$, then $f_i(w_1,\ldots,w_i,w)$ gives the probability of switching to the new edge. The behavior of $\mathcal{A}$ is completely described by these three functions. We need the following key lemma to state our construction of the adversarial input.

The lemma states (informally) that given an edge of weight $w_1$, there exists weights $x$ and $y$, close to each other such that if an edge of weight $x$ (respective $y$) is adjacent to an edge of weight $w_1$, the probability of switching is at most (respectively at least) $\delta$.

\begin{lemma}\label{lem_xy}
For every $\delta\in(0,1/\alpha)$, $\epsilon>0$, and $w_1$, there exist $x$ and $y$ such that $f_1(w_1,x)\geq\delta$, $f_1(w_1,y)\leq\delta$, $x-y\leq\epsilon$, and $w_1/\alpha\leq y\leq x\leq\gamma w_1$.
\end{lemma}
The proof of this lemma can be found in Appendix (section~\ref{sec_lemmas}).

\subsubsection{The adversarial input}

The adversarial input is parameterized by four parameters: $\delta\in(0,1/\alpha)$, $\epsilon>0$, $m$, and $n$, where $m$ and $n$ determine the graph and $\delta$ and $\epsilon$ determine the weights of its edges.

Define the infinite sequences $(x_i)_{i\in\mathbb{N}}$ and $(y_i)_{i\in\mathbb{N}}$, as functions of $\epsilon$ and $\delta$, as follows. $x_1=1$, and for all $i$, having defined $x_i$, let $x_{i+1}$ and $y_i$ be such that $f_1(x_i,x_{i+1})\geq\delta$, $f_1(x_i,y_i)\leq\delta$, $x_{i+1}-y_i\leq\epsilon$, and $x_i/\alpha\leq y_i\leq x_{i+1}\leq\gamma x_i$. Lemma \ref{lem_xy} ensures that such $x_{i+1}$ and $y_i$ exist. Furthermore, by induction on $i$, it is easy to see that for all $i$,
\begin{equation}\label{eqn_bounds}
1/\alpha^i\leq y_i\leq x_{i+1}\leq\gamma^i
\end{equation}
These sequences will be the weights of the edges in the input graph.

Given $m$ and $n$, the input graph contains several layers of vertices, namely $A_1,A_2,\dots,A_{n+1},A_{n+2}$ and $B_1,B_2,\dots,B_{n+1}$; each layer containing $m$ vertices. The vertices in the layer $A_i$ are named $a^i_1,a^i_2,\ldots,a^i_m$, and those in layer $B_i$ are named analogously. We have a complete bipartite graph $J_i$ between layer $A_i$ and $A_{i+1}$ and an edge between $a^i_j$ and $b^i_j$ for every $i$, $j$ (that is, a matching $M_i$ between $A_i$ and $B_i$).

For $i=1$ to $n$, the edges $\{(a^i_j,a^{i+1}_{j'})|1\leq j,j'\leq m\}$, in the complete bipartite graph between $A_i$ and $A_{i+1}$, have weight $x_i$, and the edges $\{(a^i_j,b^i_j)|1\leq j\leq m\}$, in the matching between $A_i$ and $B_i$, have weight $y_i$. The edges in the complete graph $J_{n+1}$ have weight $x_n$, and those in the matching $M_{n+1}$ have weight $y_n$. Note that weights $x_i$ and $y_i$ depend on $\epsilon$ and $\delta$, but are independent of $m$ and $n$. Clearly, the weight of the maximum weight matching in this graph is bounded from below by the weight of the matching $\bigcup_{i=1}^{n+1}M_i$. Since $y_i\geq x_{i+1}-\varepsilon$, we have
\begin{equation}\label{eqn_opt_lb}
\opt\geq m\left(\sum_{i=1}^ny_i+y_n\right)\geq m\left(\sum_{i=2}^{n+1}x_i+x_{n+1}-(n+1)\epsilon\right)
\end{equation}

The edges of the graph are revealed in $n+1$ phases. In the $i^{\text{\tiny{th}}}$ phase, the edges in $J_i\cup M_i$ are revealed as follows. The phase is divided into $m$ sub phases. In the $j^{\text{\tiny{th}}}$ sub phase of the $i^{\text{\tiny{th}}}$ phase, edges incident on $a^i_j$ are revealed, in the order $(a^i_j,a^{i+1}_1),(a^i_j,a^{i+1}_2),\ldots,(a^i_j,a^{i+1}_m),$ $(a^i_j,b^i_j)$. %Finally, in the $n+1^{\text{\tiny{st}}}$ phase, the edges of $M_{n+1}$ are revealed in the order $(a^{n+1}_1,b^{n+1}_1),(a^{n+1}_2,b^{n+1}_2),\ldots,(a^{n+1}_m,b^{n+1}_m)$.

\subsubsection{Analysis of the lower bound}

The overall idea of bounding the weight of the algorithm's matching is as follows. In each phase $i$, we will prove that as many as $m-O(1)$ edges of $J_i$ and only $\delta m+O(1)$ edges of $M_i$ are picked by the algorithm. 
%Since the probability of switching from an edge in $K_i$ to an edge in $M_i$ is at most $\delta$, the algorithm picks at most $\delta m+O(1)$ edges out of the $m$ edges of $M_i$. 
Furthermore, in the ${i+1}^{\text{\tiny{th}}}$ phase, since $m-O(1)$ edges from $J_{i+1}$ are picked, all but $O(1)$ edges of the edges picked from $J_i$ are discarded. Thus, the algorithm ends up with $\delta m+O(1)$ edges from each $M_i$, and $O(1)$ edges from each $J_i$, except possibly $J_n$ and $J_{n+1}$. The algorithm can end up with at most $m$ edges from $J_n\cup J_{n+1}$, since the size of the maximum matching in $J_n\cup J_{n+1}$ is $m$. Thus, the weight of the algorithm's matching is at most $mx_n$ plus a quantity that can be neglected for large $m$ and small $\delta$.

Let $X_i$ (resp. $Y_i$) be the set of edges of $J_i$ (resp. $M_i$) held by the algorithm 
%immediately after the $i^{\text{\tiny{th}}}$ phase is over. Let $X'_i$ be the set of edges of $K_i$ held by the algorithm 
at the end of input. Then we have,

\begin{lemma}\label{lem_Y}
For all $i=1$ to $n$
\[E[|Y_i|]\leq\delta m+\frac{1-\delta}{\delta}\]
\end{lemma}

\begin{lemma}\label{lem_X}
For all $i=1$ to $n-1$
\[E[|X_i|]\leq\frac{1-\delta}{\delta}\]
\end{lemma}

\begin{lemma}\label{lem_Yn}
\[E[|Y_{n+1}|]\leq\delta m+\frac{1-\delta}{\delta}\]
\end{lemma}
The proof of the above lemmas can be found in Appendix (section~\ref{sec_lemmas}).

We are now ready to prove Theorem \ref{thm_local}. The expected weight of the  matching held by $\mathcal{A}$ is
\[E[\alg]\leq\sum_{i=1}^ny_iE[|Y_i|]+y_nE[|Y_{n+1}|]+\sum_{i=1}^{n-1}x_iE[|X_i|]+x_nE[|X_n\cup X_{n+1}|]\]
Using Lemmas \ref{lem_Y}, \ref{lem_Yn}, \ref{lem_X}, and the facts that $y_i\leq x_{i+1}$ for all $i$ and $E[|X_n\cup X_{n+1}|]\leq m$ (since $X_n\cup X_{n+1}$ is a matching in $J_n\cup J_{n+1}$), we have
\[E[\alg]\leq\left(\delta m+\frac{1-\delta}{\delta}\right)\left(\sum_{i=2}^{n+1}x_i+x_{n+1}\right)+\frac{1-\delta}{\delta}\sum_{i=1}^{n-1}x_i+mx_n\]
Since the algorithm is $\beta$-competitive, for all $n$, $m$, $\delta$ and $\epsilon$ we must have $E[\alg]$ $\geq \opt/\beta$. From the above and equation (\ref{eqn_opt_lb}), we must have
\begin{center}
\begin{tabular}{ccc}
$\left(\delta m+\frac{1-\delta}{\delta}\right)\left(\sum_{i=2}^{n+1}x_i+x_{n+1}\right)$ & \multirow{2}{*}{$\geq$}& \multirow{2}{*}{$\frac{m}{\beta}\left(\sum_{i=2}^{n+1}x_i+x_{n+1}-(n+1)\epsilon\right)$}\\
$+\frac{1-\delta}{\delta}\sum_{i=1}^{n-1}x_i+mx_n$
\end{tabular}
 \end{center}
Since the above holds for arbitrarily large $m$, ignoring the terms independent of $m$ (recall that $x_i$'s are functions of $\epsilon$ and $\delta$ only), we have for all $\delta$ and $\epsilon$,
\[\delta\left(\sum_{i=2}^{n+1}x_i+x_{n+1}\right)+x_n\geq\frac{1}{\beta}\left(\sum_{i=2}^{n+1}x_i+x_{n+1}-(n+1)\epsilon\right)\]
that is,
\[x_n\geq\frac{1}{\beta}\left(\sum_{i=2}^{n+1}x_i+x_{n+1}-(n+1)\epsilon\right)-\delta\left(\sum_{i=2}^{n+1}x_i+x_{n+1}\right)\]
Taking limit inferior as $\delta\rightarrow0$ in the above inequality, and noting that limit inferior is super-additive we get for all $\epsilon$,
\begin{center}
 \begin{tabular}{ccl}
  \multirow{2}{*}{$\liminf_{\delta\rightarrow0}x_n$} & \multirow{2}{*}{$\geq$} &$\frac{1}{\beta}\left(\sum_{i=2}^{n+1}\liminf_{\delta\rightarrow0}x_i+\liminf_{\delta\rightarrow0}x_{n+1}-(n+1)\epsilon\right)$\\
  &&$-\limsup_{\delta\rightarrow0}\delta\left(\sum_{i=2}^{n+1}x_i+x_{n+1}\right)$
 \end{tabular}

\end{center}
\begin{comment}
\[\liminf_{\delta\rightarrow0}x_n\geq\frac{1}{\beta}\left(\sum_{i=2}^{n+1}\liminf_{\delta\rightarrow0}x_i+\liminf_{\delta\rightarrow0}x_{n+1}-(n+1)\epsilon\right)-\limsup_{\delta\rightarrow0}\delta\left(\sum_{i=2}^{n+1}x_i+x_{n+1}\right)\]
\end{comment}
Recall that $x_i$'s are functions of $\epsilon$ and $\delta$, and that from equation (\ref{eqn_bounds}), $1/\alpha^i\leq x_{i+1}\leq \gamma^i$, where the bounds are independent of $\delta$. Thus, all the limits in the above inequality exist. Moreover, $\lim_{\delta\rightarrow0}\delta\left(\sum_{i=2}^{n+1}x_i+x_{n+1}\right)$ exists and is $0$, for all $\epsilon$. This implies $\limsup_{\delta\rightarrow0}\delta\left(\sum_{i=2}^{n+1}x_i+x_{n+1}\right)=0$ and we get for all $\varepsilon$,
\[\liminf_{\delta\rightarrow0}x_n\geq\frac{1}{\beta}\left(\sum_{i=2}^{n+1}\liminf_{\delta\rightarrow0}x_i+\liminf_{\delta\rightarrow0}x_{n+1}-(n+1)\epsilon\right)\]
Again, taking limit inferior as $\epsilon\rightarrow0$, and using super-additivity,
\[\liminf_{\epsilon\rightarrow0}\liminf_{\delta\rightarrow0}x_n\geq\frac{1}{\beta}\left(\sum_{i=2}^{n+1}\liminf_{\epsilon\rightarrow0}\liminf_{\delta\rightarrow0}x_i+\liminf_{\epsilon\rightarrow0}\liminf_{\delta\rightarrow0}x_{n+1}\right)\]
Note that the above holds for all $n$. Finally, let $\overline{x_n}=\liminf_{\epsilon\rightarrow0}\liminf_{\delta\rightarrow0}x_{n+1}$. Then we have the infinite sequence $(\overline{x_n})_{n\in\mathbb{N}}$ such that for all $n$, $\beta\overline{x_n}\geq\sum_{i=1}^{n+1}\overline{x_i}+\overline{x_{n+1}}$. Thus, by Lemma \ref{lem_varadaraja}, we have $\beta\geq3+2\sqrt{2}$.

\subsection{Lower Bound for $\theta$ structured graphs}
Recall that an edge weighted graph is said to be $\theta$-structured if the weights of the edges are powers of $\theta$. The following bound applies to any deterministic algorithm for MWM on $\theta$-structured graphs.
\begin{theorem}\label{thm_theta}
No deterministic algorithm  can have a competitive ratio less than $2+\frac{2}{\theta-2}$ for MWM on $\theta$-structured graphs, for $\theta\geq 4$.
\end{theorem}
The proof of the above theorem can be found in Appendix (section~\ref{theta_bnd}).

 \section{Randomized Algorithm for Paths}\label{ub}
 
 When the input graph is restricted to be a collection of paths, then every new edge that arrives connects two (possibly empty) paths. Our algorithm consists of several cases, depending on the lengths of the two paths. 
 \begin{algorithm}[H]
  \caption{Randomized Algorithm for Paths}
  \begin{algorithmic}[1]
   \STATE $M=\emptyset$. \COMMENT{$M$ is the matching stored by the algorithm.}
   \FOR{each new edge $e$}
    \STATE Let $L_1\geq L_2$ be the lengths of the two (possibly empty) paths $P_1,P_2$ that $e$ connects.
    \STATE If $L_1>0$ (resp. $L_2>0$), let $e_1$ (resp. $e_2$) be the edge on $P_1$ (resp. $P_2$) adjacent to $e$.
    \IF{$e$ is a disjoint edge \COMMENT{$L_1=L_2=0$ } } 
    \STATE{$M=M\cup \{e\}$.}
    \ELSIF{$e$ is revealed on a disjoint edge $e_1$ \COMMENT{$L_1=1,L_2=0$. $e_1\in M$}}  \STATE{with probability $\frac{1}{2}$, $M=M\setminus\{e_1\}\cup\{e\}$}.
    \ELSIF{$e$ is revealed on a end point of path of length $>1$ \COMMENT{$L_1>1,L_2=0$}} \STATE{if $e_1\notin M$, $M=M\cup \{e\}$ }.
    \ELSIF{$e$ joins two disjoint edges \COMMENT{$L_1=L_2=1$. $e_1,e_2\in M$}} \STATE{with probability $\frac{1}{2}$, $M=M\setminus \{e_1,e_2\}\cup \{e\}$}.
    \ELSIF{$e$ joins a path and a disjoint edge \COMMENT{$L_1>1,L_2=1$. $e_2\in M$}} \STATE{if $e_1\notin M$, $M=M\setminus\{e_2\}\cup\{e\}$}.
    \ELSIF{$e$ joins two paths of length $>1$\COMMENT{$L_1>1,L_2>1$}} \STATE{if $e_1\notin M$ and $e_2\notin M$, $M=M\cup \{e\}$}.
    \ENDIF
    \STATE Output $M$.
    \ENDFOR
   \end{algorithmic}
 \end{algorithm}

The following simple observations can be made by looking at the algorithm:
\begin{itemize}
 \item All isolated edges belong to $M$ with probability one.
 \item The end vertex of any path of $length>1$ is covered by $M$ with probability $\frac{1}{2}$, and this is independent of the end vertex of any other path being covered.
 \item For a path of length $2,3,$ or $4$, each maximal matching is present in $M$ with probability $\frac{1}{2}$.
\end{itemize}

 \begin{theorem}\label{th1}
  The randomized algorithm for finding MCM on path graphs is $\frac{4}{3}$-competitive.
 \end{theorem}
The proof of above theorem can be found in Appendix (section~\ref{opt_paths}).
\begin{comment}
\section{Conclusion}
In this paper, we present a barely random algorithm for MCM on growing trees. The next step, is to use the same algorithm to prove a competitive ratio better than $2$ for trees, where edges may be revealed in any arbitrary order. The deterministic algorithm due to McGregor~\cite{mcgregor} uses a single parameter $\gamma$. It will be interesting to see if there exists a barely random algorithm that performs provably better than the best known randomized algorithm~\cite{epstein} for MWM, which chooses between output of two algorithms with some probability, where both the algorithms use different values for the parameter $\gamma$.

\section*{Acknowledgment}
The first author thanks Sagar Kale for reviewing Section \ref{sec_lb_mwm}, correcting several mistakes and helping us improve the presentation substantially.
\end{comment}
\bibliographystyle{plain}% the recommended bibstyle
\bibliography{paper}
~\nocite{*}
\newpage
\section*{Appendices}
\appendix
% \appendixpage

 \section{A Primal-Dual Analysis of a deterministic algorithm for MWM}\label{pd}
 In this section, we present a primal-dual analysis for the deterministic algorithm due to~\cite{mcgregor} for the maximum weight matching problem in the online preemptive model. The algorithm is as follows.
 \begin{algorithm}
  \caption{Deterministic Algorithm for MWM}
  \begin{enumerate}
   \item Fix a parameter $\gamma$.
   \item If the new edge $e$ has weight greater than $(1+\gamma)$ times the weight of the edges currently adjacent to $e$, then include $e$ and discard the adjacent edges.
  \end{enumerate}

 \end{algorithm}
\subsection{Analysis}
\begin{lemma}~\cite{mcgregor}
The competitive ratio of this algorithm is $(1+\gamma)(2+\frac{1}{\gamma})$. 
\end{lemma}
We use the primal-dual technique to prove the same competitive ratio. This analysis technique is different from the one in~\cite{naor}. In~\cite{naor}, the primal variables once set to a certain value are never changed whereas in our analysis the primal variables may change during the run of algorithm. The primal and dual LPs we use for the maximum weight matching problem are as follows.
 \begin{center}
 \begin{tabular}{c|c}
 Primal LP & Dual LP \\\hline
 $\max \sum_e w_e x_e$ & $\min \sum_v y_v$\\
 $\forall v: \sum_{v\in e}x_e \leq 1$ & $\forall e: y_u + y_v \geq w_e$\\
 $x_e\geq 0$ & $y_v \geq 0$
 \end{tabular}
\end{center}
We maintain both primal and dual variables along with the run of the algorithm. On processing an edge, we maintain the following invariants.
\begin{itemize}
 \item The dual LP is always feasible.
 \item For each edge $e\equiv(u,v)$ in the current matching, $y_u\geq(1+\gamma)w(e)$ and $y_v\geq(1+\gamma)w(e)$.
 \item The change in cost of the dual solution is at most $(1+\gamma)(2+\frac{1}{\gamma})$ times the change in cost of primal solution.
\end{itemize}
 These invariants imply that the competitive ratio of the algorithm is $(1+\gamma)(2+\frac{1}{\gamma})$.

We start with $\vec{0}$ as the initial primal and dual solutions. Consider a round in which an edge $e$ of weight $w$ is given. Assume that all the above invariants hold before this edge is given. Whenever an edge $e\equiv(u,v)$ is accepted by the algorithm, we assign values $x_e=1$ to the primal variable and $y_u=max(y_u,(1+\gamma)w(e))$, $y_v=max(y_v,(1+\gamma)w(e))$ to the dual variables of its end points. Whenever an edge is rejected, we do not change the corresponding primal or dual variables. Whenever an edge $e$ is evicted, we change its primal variable $x_e=0$. The dual variables never decrease. Hence, if a dual constraint is feasible once, it remains so. We will now show that the invariants are always satisfied. These are three cases.
\begin{enumerate}
 \item If the edge $e\equiv(u,v)$ has no conflicting edges in the current matching, then it is accepted by the algorithm in current matching $M$. We assign $x_e=1,y_u=\max(y_u,(1+\gamma)w(e))$ and $y_v=\max(y_v,(1+\gamma)w(e))$. Hence, $y_u\geq (1+\gamma)w(e)$ and $y_v\geq (1+\gamma)w(e)$. And hence, the dual constraint $y_u+y_v\geq w(e)$ is feasible. The change in the dual cost is at most $2(1+\gamma)w(e)$. The change in the primal cost is $w(e)$. So, the change in the dual cost is at most $(1+\gamma)(2+\frac{1}{\gamma})$ times the change in the cost of the primal solution.
 \item If the edge $e\equiv(u,v)$ has conflicting edges $X(M,e)$ and $w(e)\leq (1+\gamma)w(X(M,e))$, then it is rejected by the algorithm. That happens when $y_u+y_v\geq (1+\gamma)w(X(M,e))$, and the dual constraint for edge $e$ is satisfied.
 \item If the edge $e\equiv(u,v)$ had conflicting edges $X(M,e)$ and $w(e)> (1+\gamma)w(X(M,e))$, then it is accepted by the algorithm in the current matching $M$ and $X(M,e)$ is/are evicted from $M$. We only need to show that the change in dual cost is at most $(1+\gamma)(2+\frac{1}{\gamma})$ times the change in the primal cost. The change in primal cost is $w(e)-w(X(M,e))$. The change in dual cost is at most $2(1+\gamma)w(e)-(1+\gamma)w(X(M,e))$. Hence the ratio is at most
 \begin{align*}
	      & \frac{2(1+\gamma)w(e)-(1+\gamma)w(X(M,e))}{w(e)-w(X(M,e))}\\
              =&2(1+\gamma) + \frac{(1+\gamma)w(X(M,e))}{w(e)-w(X(M,e))}\\
              <&2(1+\gamma) + \frac{w(e)}{w(e)-w(X(M,e))}\\
              \leq &2(1+\gamma)+(1+\frac{1}{\gamma})\\
              =&(1+\gamma)(2+\frac{1}{\gamma})
 \end{align*}

\end{enumerate}
Here, the management of the dual variables was straight forward. The introduction of randomization complicates matters considerably, and we are only able to analyze the algorithm in the very restricted setting of paths and ``growing trees''.

\section{Barely Random Algorithms for MCM}\label{ra_paths}

\subsection{Randomized Algorithm for Paths}
 \begin{algorithm}
  \caption{Barely Random Algorithm for Paths}
  \begin{enumerate}
    \item The algorithm maintains two matchings: $M_1$ and $M_2$.
    \item On receipt of an edge $e$, the processing happens in two phases.
   \begin{enumerate}
    \item The augment phase. Here, the new edge $e$ is added to each $M_i$ such that there is no edge in $M_i$ sharing an end point with $e$.
    \item The switching phase. Edge $e$ is added to $M_2$ and the conflicting edge is discarded, provided it decreases the quantity $|M_1\cap M_2|$.
   \end{enumerate}
   \item Output a matching $M_i$ with probability $\frac{1}{2}$.
   \end{enumerate}
 \end{algorithm}
\begin{theorem}\label{br_paths}
The barely random algorithm for finding the MCM on paths is $\frac{3}{2}$-competitive, and no barely random algorithm can do better.
\end{theorem}
%\begin{theorem}
  %There exists a barely random algorithm for finding the MCM on growing trees with maximum degree $3$ which is $\frac{12}{7}$-competitive.	
 %\end{theorem}
We prove this theorem using the following lemma.
\begin{lemma}
The dual constraint for each edge is satisfied at least $\frac{2}{3}rd$ in expectation.
\end{lemma}
$M_1$ and $M_2$ are valid matchings and hence correspond to valid primal solutions. For each edge $e\equiv(u,v)$ in some matching $M_i$, we distribute a charge of $x_e=1$ amongst dual variables $y_u$ and $y_v$ of its vertices. We prove that for each edge $e$, $y_u+y_v\geq \frac{4}{3}$. Thus, $\mathbb{E}[y_u+y_v]\geq \frac{2}{3}$. Hence, this algorithm has a competitive ratio $\frac{3}{2}$. All the dual variables are initialized to $0$. Suppose $e\equiv(u,v) \in M_i$ for some $i\in[2]$. Then distribution of primal charge $x_e$ amongst dual variables $y_u$ and $y_v$ is done as follows. If there is an edge incident on $u$ which does not belong to any matching, and there is an edge incident on $v$ which does belong to some matching, then $e$ transfer a primal charge of $\frac{2}{3}$ to $y_u$ and rest is transferred to $y_v$. Else, the primal charge of $e$ is transferred equally amongst $y_u$ and $y_v$.

We look at three cases and prove that $y_u+y_v\geq \frac{4}{3}$ for each edge $e\equiv (u,v)$.
\begin{enumerate}
 \item The edge $e$ is not present in any matching.
 \begin{enumerate}
  \item If there are no edges on both its end points in the input graph, then this edge has to be covered by both $M_1$ and $M_2$. So this case is not possible.
  \item If there is no edge on one end point (say $u$) in the input graph, then $e$ has to belong to belong to some matching. So, this case is not possible.
  \item If there are two edges incident on end points of $e$ in the input graph, then each of them has to be covered by some matching. So, $y_u+y_v\geq 2\cdot\frac{2}{3}=\frac{4}{3}$.
 \end{enumerate}
\item The edge $e\equiv(u,v)$ is present in a single matching.
 \begin{enumerate}
  \item If there are no edges on both its end points in the input graph, then this edge has to be covered by both $M_1$ and $M_2$. So this case is not possible.
  \item If there is no edge on one end point (say $u$) in the input graph, then edge on the other end point must be covered by the other matching. Otherwise, edge $e$ would have been covered by both matchings. So, $y_u+y_v\geq  1 +\frac{1}{3}=\frac{4}{3}$.
  \item If there two edges incident on end points of $e$ in the input graph, then at least one of them has to be covered by other matching. Else, edge $e$ would have been covered by both matchings. So, $y_u+y_v\geq  1 +\frac{1}{3}=\frac{4}{3}$.
 \end{enumerate}
\item The edge $e\equiv(u,v)$ is present in both the matchings, then $y_u+y_v=2\geq \frac{4}{3}$.
\end{enumerate}
This proves the above claim. The corollary of the above claim is that we have a $\frac{3}{2}$-competitive randomized algorithm for the MCM on paths.

\begin{proof} (of the second part of Theorem~\ref{br_paths}) 
Suppose $U$ is the set of matchings used by a barely random algorithm $\mathcal{A}$. Following input is given to this algorithm. Reveal two edges $x_1$ and $y_1$ such that they share an end point. Let $S$ be the set of matchings to which $x_1$ is added, and $\bar{S}$ be the set of matchings to which $y_1$ is added. Here, $U=S\cup \bar{S}$. Now give two more edges $x_2$ and $y_2$ disjoint from the previous edges, such that $x_2$ and $y_2$ share an end point. Wlog, $x_2$ will be added to set of matchings $S$, and $y_2$ will be added to set of matchings $\bar{S}$. Give an edge between $y_2$ and $x_1$. Continue the input similarly for $i>2$. It can be seen that expected increase in the size of optimum matching is $\frac{3}{2}$, whereas increase in the size of matching held by the algorithm is $1$. Thus, we get a lower bound $\frac{3}{2}$ on the competitive ratio of any barely random algorithm.
 
\end{proof}

\subsection{Randomized Algorithm for Growing Trees with maximum degree $3$}\label{ub2}
 In this section, we give a barely random algorithm for growing trees, with maximum degree $3$. We beat the lower bound of $2$ for MCM on the performance of any deterministic algorithm, for this class of inputs. The edges are revealed in online fashion such the one new vertex is revealed per edge, (except for the first edge). Any vertex in the input graph has maximum degree $3$.

  \begin{algorithm}[H]
  \caption{Randomized Algorithm for Growing Trees with $\Delta=3$}
  \begin{enumerate}
    \item The algorithm maintains $3$ matchings $M_1,M_2,M_3$.
    \item On receipt of an edge $e$, the processing happens in two phases.
   \begin{enumerate}
    \item The augment phase. Here, the new edge $e$ is added to each $M_i$ such that there is no edge in $M_i$ sharing an end point with $e$.
    \item The switching phase. For $i=2,3$, in order, $e$ is added to $M_i$ and the conflicting edge is discarded, provided it decreases the quantity $\sum_{i,j\in[3],i\neq j}|M_i\cap M_j|$.
   \end{enumerate}
   \item Output a matching $M_i$ with probability $\frac{1}{3}$.
   \end{enumerate}
 \end{algorithm}
 
\begin{theorem}\label{br_trees}
 The barely random algorithm for finding the MCM on growing trees with maximum degree $3$ is $\frac{12}{7}$-competitive.	
\end{theorem}
 
We make following simple observations.
\begin{itemize}
 \item There cannot be an edge which is not in any matching.
 \item Call an edge ``bad'' if its end points are covered by only two matchings. Indeed, an edge whose none of the end points are leaves, cannot be ``bad''.
 \item An edge incident on a vertex of degree $3$ cannot be ``bad'', because there will be a distinct edge belonging to every matching.
 \end{itemize}
 We begin by proving a few simple lemmas regarding the algorithm.
 \begin{lemma}
 There cannot be ``bad'' edges incident on both vertices of an edge.
 \end{lemma}
 \begin{proof}
  Note that a ``bad'' edge is created when an edge is revealed on a leaf node of another edge which belongs to two or three matchings currently and finally belongs to only one matching.
  
  Let $e$ be an edge which currently belongs to three matchings, which means it is the first edge revealed. Now if an edge $e_1$ is revealed on a vertex of $e$, then $e_1$ would be added to one matching, and $e$ would be removed from that matching, (in the switching phase of the algorithm). For $e_1$ to be a bad edge, $e$ should be switched out of one more matching. This can happen only if there are two more edges revealed on the other vertex of $e$. This means there cannot be ``bad'' edges on both sides of $e$.
  
  Let $e$ belongs to two matchings. Then $e$ already has a neighboring edge $e_2$ which belongs to some matching. When $e_1$ is revealed on the leaf vertex of $e$, it will be added to one matching, in the augment phase. Now for $e_1$ to be ``bad'', $e$ should switch out of some matching. This can only happen if there is one more edge $e_3$ revealed on the common vertex of $e$ and $e_2$. Again, the lemma holds.
 \end{proof}

 \begin{lemma}\label{3star}
 If a vertex has three edges incident on it, then at most one of these edges can have a ``bad'' neighboring edge.
 \end{lemma}
\begin{proof}
 Out of the three edges incident on a vertex, only one could have belonged to two matchings at any step during the run of algorithm. Hence, only that edge which belonged to two matchings at some stage during the run of algorithm can have a ``bad'' neighboring edge.
\end{proof}

 \begin{proof} (of Theorem~\ref{br_trees})
$M_1, M_2, M_3$ are valid matchings and hence correspond to valid primal solutions. For each edge $e\equiv(u,v)$ in some matching $M_i$, we distribute a charge of $x_e=1$ amongst dual variables $y_u$ and $y_v$ of its end points. We prove that for each edge $e$, $y_u+y_v\geq \frac{7}{4}$. Thus, $\mathbb{E}[y_u+y_v]\geq \frac{7}{12}$. Hence, this algorithm has a competitive ratio $\frac{12}{7}$. All the dual variables are initialized to $0$. Suppose $e\equiv(u,v) \in M_i$ for some $i\in[3]$. Then distribution of primal charge $x_e$ amongst dual variables $y_u$ and $y_v$ is done as follows. If there is a ``bad'' edge incident on $u$, then edge $e$ transfer $\frac{3}{4}$ of of its primal charge to $y_u$ and rest of it to $y_v$. Else, edge $e$ transfer its primal charge equally between $y_u$ and $y_v$.

We look at three cases and then prove that $y_u+y_v\geq \frac{7}{4}$ for each edge $e\equiv (u,v)$.
\begin{enumerate}
 \item Edge $e\equiv(u,v)$ is ``bad''. $e\in M_i$ for some $i\in [3]$. $e$ will have some neighboring edge $e_1$ such that $e_1\in M_j$ for $j\in [3]$ and $i\neq j$. Let the common vertex between $e$ and $e_1$ be $v$. Then $e_1$ will transfer $\frac{3}{4}$ of its primal charge to $y_v$. Thus, $y_u+y_v=\frac{7}{4}$.
 \item Edge $e\equiv(u,v)$ is present in a single matching and not ``bad''. This case has four sub cases.
 \begin{enumerate}
  \item $e$ has one neighboring edge $e_1$. Then $e_1$ should belong to two matchings.
  \item $e$ has two neighboring edges $e_1$ and $e_2$ both belonging to only one matching. If these are both on the same side of $e$, then at most one of them could have a ``bad'' neighboring edge (by lemma~\ref{3star}). If these are on opposite sides of $e$, then none of them can have a ``bad'' neighboring edge.
  \item $e$ has three neighboring edges $e_1$, $e_2$, and $e_3$, such that $e_1$ and $e_2$ are on one side of $e$, and $e_3$ is on another side of $e$. At most one of $e_1$ and $e_2$ can have a ``bad'' neighboring edge (by lemma~\ref{3star}).
  \item $e$ has four neighboring edges $e_1$, $e_2$, $e_3$, and $e_4$, such that $e_1$ and $e_2$ are on one side of $e$, and $e_3$ and $e_4$ are on another side of $e$.
 \end{enumerate}
 We can see that in all the above sub cases, $y_u+y_v\geq \frac{7}{4}$.
 \item Edge $e\equiv (u,v)$ belongs to two or three matchings. Then, $y_u+y_v\geq \frac{7}{4}$ trivially.
\end{enumerate}
This proves that we have a $\frac{12}{7}$-competitive randomized algorithm for finding MCM on growing trees with maximum degree $3$.
\end{proof}

\subsection{Example showing need of non-local analysis}\label{sa_gt}
Consider input graph as a $4$-regular tree with large number of vertices, and an extra edge on every vertex other than the leaf vertices. Every edge other than the extra edges will belong some matching. For every edge that belongs to some matching, there will one edge on each of its end points which does not belong to any matching. If the rule for distributing primal charge among dual variables is similar to one described in section~\ref{ub2}, then for each edge belonging to some matching will transfer its primal charge equally amongst both its end points. For each edge which does not belong to any matching, $y_u+y_v=2$, which will imply only a competitive ratio of $2$. We wish to get a competitive ratio better than $2$. So we need some other idea.

\subsection{Proof of Lemma~\ref{internal}}\label{pf_inbad}
\begin{proof} Consider an edge $(u,v)$ revealed at $u$. 
\begin{enumerate}
 \item When revealed it is not put in any matching. This means that there are four covered edges incident on $u$. (Call an edge covered if it belongs to some matching.) This situation cannot change as more edges are revealed. Thus the edge will remain covered by four matchings, and can never become a bad edge.
 \item When revealed it is put in one matching. This means that there are three matching edges on at least two covered edges incident on $u$. If there were three covered edges incident on $u$ then they remain covered edges. So suppose otherwise. Then there are two covered edges of which one is in two matchings. Hence there will always be three matching edges covering $u$. If an edge is revealed at $v$ then there will be four matching edges covering the given edge. The edge may become bad if $v$ stays a leaf and  if one of the matchings on the edge with two of them, switches.
 \item When revealed it is put in two matchings. Then there are two matching edges at $u$ and at least one covered edge. If there are two covered edges, they remain so. Of the two copies of the edge in matchings, one may switch to a new edge but will always remain adjacent to this edge. Hence there will always be three matching edges  covering $u$. If an edge is revealed at $v$ then there will be four matching edges covering the given edge. The edge may become bad if $v$ stays a leaf and  if one of the matchings on the edge with two of them, switches.
 \item When revealed it is put in three matchings. Then there is one covered edge at $u$. If one more edge is now revealed on $u$, then we are back to case $3$. If a new edge is revealed on $v$, it replaces $(u,v)$ in one of the matchings. Now, even if more edges are revealed on either side of $(u,v)$, it continues to be covered by four matchings.
 \item When revealed it is put in four matchings. If a new edge is revealed either on $u$ or $v$, then this case reduces to case $2$.
\end{enumerate}
This completes the proof of the first part of lemma.

 For the second part of lemma, consider a leaf edge present on each of the vertices $p,q,$ and $r$. Suppose the leaf edge incident on $q$ is bad. When this edge was revealed, there must have been some edge incident on $q$, either $(p,q)$ or $(q,r)$, which belonged to two matchings. Wlog, assume $(p,q)$ belonged to two matchings. Then for a matching to switch out this edge, there need to be three edges incident on $p$, and hence the leaf edge incident on $p$ cannot be a bad edge.

\end{proof}

\section{Proof of lemmas from section~\ref{sec_lb_mwm}}\label{sec_lemmas}

\begin{lemma}\label{lem_f0}
For every $w>0$, $f_0(w)>1/\alpha$.
\end{lemma}
\begin{proof}
If not, then a single edge of weight $w$ results in algorithm's expected cost $wf_0(w)\leq w/\alpha<w/\beta$, whereas the optimum is $w$. This contradicts $\beta$-competitiveness.
\end{proof}

\begin{lemma}\label{lem_f1_lb}
For every $w_1$ and $w\leq w_1/\alpha$, $f_1(w_1,w)=0$.
\end{lemma}
\begin{proof}
If $f_1(w_1,w)>0$ for some $w_1$ and $w$ such that $w\leq w_1/\alpha$, then the adversary's input is a star, with a single edge of weight $w_1$ followed by a large number $n$ of edges of weight $w$. Regardless of whether the first edge of weight $w_1$ is accepted or not, the algorithm holds an edge of weight $w$, with probability approaching $1$ as $n\rightarrow\infty$, in the end. The optimum is $w_1\geq\alpha w>\beta w$, thus, contradicting $\beta$-competitiveness.
\end{proof}

\begin{lemma}\label{lem_f1_ub}
For every $w_1$, and $w\geq\gamma w_1$, $f_1(w_1,w)\geq1/\alpha$.
\end{lemma}
\begin{proof}
Suppose $f_1(w_1,w)<1/\alpha$ for some $w_1$ and $w$ such that $w\geq\gamma w_1$. The adversary's input is a star, with a large number $n$ of edges of weight $w_1$, followed by a single edge of weight $w$. The algorithm must hold an edge of weight $w_1$, before the edge of weight $w$ is given, with probability approaching $1$ as $n\rightarrow\infty$. Therefore, in the end, the algorithm's cost is $w$ with probability less than $1/\alpha$, and at most $w_1$ otherwise. Thus, the expected weight of the edge held by the algorithm is less than $w/\alpha+w_1(1-1/\alpha)$, whereas the adversary holds the edge of weight $w$. Since the algorithm is $\beta$-competitive and $\beta<\alpha$, we have
\[w<\beta\left[\frac{1}{\alpha}\cdot w+\left(1-\frac{1}{\alpha}\right)w_1\right]\text{ }\Rightarrow\text{ }\left(1-\frac{\beta}{\alpha}\right)w<w_1\beta\left(1-\frac{1}{\alpha}\right)\text{ }\Rightarrow\text{ }w<\gamma w_1\]
This is a contradiction.
\end{proof}

\begin{proof}[Proof of Lemma~\ref{lem_xy}]
 By Lemma \ref{lem_f1_lb}, $f_1(w_1,w_1/\alpha)=0$, and by Lemma \ref{lem_f1_ub}, $f_1(w_1,$ $\gamma w_1)\geq1/\alpha$. Take a finite sequence of points, increasing from $w_1/\alpha$ to $\gamma w_1$, such that the difference between any two consecutive points is at most $\epsilon$, and observe the value of $f_1(w_1,z)$ at each such point $z$. Since $f_1(w_1,w_1/\alpha)<\delta$ and $f_1(w_1,\gamma w_1)>\delta$, there must exist two consecutive points in the sequence, say $y$ and $x$, such that $f_1(w_1,y)\leq\delta$ and $f_1(w_1,x)\geq\delta$. Furthermore, $x-y\leq\epsilon$ and $w_1/\alpha\leq y\leq x\leq\gamma w_1$, by construction.
\end{proof}

\begin{lemma}\label{lem_preempt}
For every $i$, $j$, the probability that $a^i_j$ is not matched to any vertex in $A_{i+1}$, in the $j^{\text{\tiny{th}}}$ sub phase of the $i^{\text{\tiny{th}}}$ phase, just before the edge $(a^i_j,b^i_j)$ is revealed, is at most $(1-\delta)^{m-j+1}$.
\end{lemma}
\begin{proof}
Consider the $j^{\text{\tiny{th}}}$ sub phase of the $i^{\text{\tiny{th}}}$ phase, in which, the edges $(a^i_j,a^{i+1}_1),$ $(a^i_j,a^{i+1}_2),$ $\ldots,(a^i_j,a^{i+1}_m),$ $(a^i_j,b^i_j)$ are revealed. Before this sub phase, the number of unmatched vertices in $A_{i+1}$ must be at least $m-j+1$. Call this set $A'$. If $a^i_j$ was matched at the end of phase $i-1$, then the weight of edge incident on $a^i_j$, at the beginning of the current phase, is $x_{i-1}$. For each vertex $a^{i+1}_{j'}\in A'$, given that $a^i_j$ did not get matched to any of $a^{i+1}_1,\ldots,a^{i+1}_{j'-1}$, the probability that $a^i_j$ gets matched to $a^{i+1}_{j'}$ is $f_1(x_{i-1},x_i)\geq\delta$. Thus, the probability of $a^i_j$ not getting matched to any vertex in $A'\subseteq A_{i+1}$, in the current sub phase, is at most $(1-\delta)^{m-j+1}$. Note that this argument applies even if $a^i_j$ was not matched at the beginning of the current phase, due to Lemma \ref{lem_f0} and since $\delta<1/\alpha<f_0(x_i)$.
\end{proof}

\begin{proof} [Proof of Lemma~\ref{lem_Y}]
First, observe that the sequence in which the edges are revealed ensures that no edge adjacent to any edge $e\in M_i$ appears after $e$. Thus, if $e$ is picked when it is revealed, it is never preempted, and is maintained till the end of input. Hence, $Y_i$ is also the set of edges of $M_i$ that were picked as soon as they were revealed.

When the edge $(a^i_j,b^i_j)$ is given, the algorithm picks it with probability at most $\delta$ (since $f_1(x_i,y_i)\leq\delta$) if $a^i_j$ was matched to some vertex in $A_{i+1}$. By Lemma \ref{lem_preempt}, the probability of $a^i_j$ not being matched to any vertex in $A_{i+1}$ is at most $(1-\delta)^{m-j+1}$. Thus, the probability that the edge $(a^i_j,b^i_j)$ appears in $Y_i$ is at most $\delta+(1-\delta)^{m-j+1}$. Hence, $E[|Y_i|]\leq\delta m+\sum_{j=1}^m(1-\delta)^{m-j+1}\leq\delta m+(1-\delta)/\delta$.
\end{proof}

\begin{proof} [Proof of Lemma~\ref{lem_X}]
Consider the set $A'$ of all vertices $a^{i+1}_j$, which remain matched to some vertex in $A_i$ at the end of input. Then clearly, $|A'|=|X_i|$. Let us find the probability that a vertex $a^{i+1}_j$ appears in $A'$. For this to happen, it is necessary that $a^{i+1}_j$ not be matched to any vertex in $A_{i+2}$, in the $j^{\text{\tiny{th}}}$ sub phase of the $i+1^{\text{\tiny{st}}}$ phase. By lemma \ref{lem_preempt}, this happens with probability at most $(1-\delta)^{m-j+1}$. Thus, $E[|X_i|]=\sum_{j=1}^m(1-\delta)^{m-j+1}\leq(1-\delta)/\delta$.
\end{proof}

\begin{lemma}\label{lem_preempt1}
For every $j$, the probability that $a^{n+1}_j$ is not matched to any vertex in $A_n\cup A_{n+2}$, in the $j^{\text{\tiny{th}}}$ sub phase of the $n+1^{\text{\tiny{st}}}$ phase, just before the edge $(a^{n+1}_j,b^{n+1}_j)$ is revealed, is at most $(1-\delta)^{m-j+1}$.
\end{lemma}
\begin{proof}
This proof is analogous to the proof of Lemma \ref{lem_preempt}. If $a^{n+1}_j$ was matched to some vertex in $A_n$ at the end of the $n^{\text{\tiny{th}}}$ phase, then it will continue to remain matched to some vertex in $A_n\cup A_{n+2}$, until the edge $(a^{n+1}_j,b^{n+1}_j)$ is revealed. Otherwise $a^{n+1}_j$ will get matched to some vertex in $A_{n+2}$ with probability at least $1-(1-\delta)^{m-j+1}$, and remain unmatched with probability at most $(1-\delta)^{m-j+1}$.
\end{proof}

\begin{proof} [Proof of Lemma~\ref{lem_Yn}]
This proof is analogous to the proof of Lemma \ref{lem_Y}. Again, the sequence in which the edges are revealed ensures that no edge adjacent to any edge $e$ in any $M_{n+1}$ appears after $e$. Thus, if $e$ is picked when it is revealed, it is never preempted. Hence, $Y_i$ is also the set of edges of $M_i$ that were picked as soon as they were revealed.

When the edge $(a^{n+1}_j,b^{n+1}_j)$ is given, the algorithm picks it with probability at most $\delta$ (since $f_1(x_n,y_n)\leq\delta$) if $a^{n+1}_j$ was matched to some vertex in $A_n\cup A_{n+1}$. Thus, the probability that the edge $(a^i_j,b^i_j)$ appears in $Y_{n+1}$ is at most $\delta+(1-\delta)^{m-j+1}$. Hence, $E[|Y_{n+1}|]\leq\delta m+\sum_{j=1}^m(1-\delta)^{m-j+1}\leq\delta m+(1-\delta)/\delta$.
\end{proof}

\section{Lower Bound for $\theta$ structured graphs}\label{theta_bnd}
The overall idea of the adversarial strategy is as follows. The input graph is a tree whose edges are partitioned into $n+1$ layers which are numbered $0$ through $n$ from bottom to top. Every edge in layer $i$ has weight $\theta^i$. The edges are revealed bottom-up. The edges in layer $i$ are given in such a manner that all the edges in layer $i-1$ held by the algorithm will be preempted. This ensures that in the end, the algorithm's matching contains only one edge, whereas the adversary's matching contains $2^{n-i}$ edges from layer $i$, for each $i$. 

Let $\mathcal{A}$ be any deterministic algorithm for maximum matching in the online preemptive model. The adversarial strategy uses a recursive function, which takes $n\in\mathbb{N}$ as a parameter. For a given $n$, this recursive function, given by Algorithm \ref{alg_theta}, constructs a tree with $n+1$ layers by giving weighted edges to the algorithm in an online manner, and returns the tree, the adversary's matching in the tree, and a vertex from the tree.

\begin{algorithm}
\caption{{\sc{MakeTree}}$(n)$}
\begin{algorithmic}[1]
\IF{$n=0$}
\WHILE{true}
\STATE Take fresh vertices $v$, $v_1$, $v_2$, and give the edges $(v_1,v_2)$, $(v,v_1)$ with weight $1$.
\STATE $T:=\{(v_1,v_2),(v,v_1)\}$.
\IF{algorithm picks $(v_1,v_2)$} \RETURN $(T,\{(v,v_1)\},v_2)$
\ELSIF{algorithm picks $(v,v_1)$} \RETURN $(T,\{(v_1,v_2)\},v)$
\ELSE \STATE Discard $T$ and retry.
\ENDIF
\ENDWHILE
\ELSE
\WHILE{true}
\STATE $(T_1,M_1,v_1)$ $:=$ {\sc{MakeTree}}$(n-1)$
\STATE $(T_2,M_2,v_2)$ $:=$ {\sc{MakeTree}}$(n-1)$
\STATE Give the edge $(v_1,v_2)$ with weight $\theta^n$.
\IF{algorithm picks $(v_1,v_2)$}
\STATE Take a fresh vertex $v$, and give the edge $(v,v_1)$ with weight $\theta^n$.
\STATE $T:=T_1\cup T_2\cup\{(v_1,v_2),(v,v_1)\}$.
\IF{algorithm replaces $(v_1,v_2)$ by $(v,v_1)$} \RETURN $(T,M_1\cup M_2\cup\{(v_1,v_2)\},v)$.
\ELSE \RETURN $(T,M_1\cup M_2\cup\{(v,v_1)\},v_2)$
\ENDIF
\ELSE 
\STATE \COMMENT{algorithm does not pick $(v_1,v_2)$}
\STATE Discard the constructed tree and retry. 
\ENDIF
\ENDWHILE
\ENDIF
\end{algorithmic}
\label{alg_theta}
\end{algorithm}

Let us prove a couple of properties about the behavior of the algorithm and the adversary, when the online input is generated by the call {\sc{MakeTree}}$(n)$.

\begin{lemma}\label{lem_theta_adv}
Suppose that the call {\sc{MakeTree}}$(n)$ returns $(T',M',v')$. Then
\begin{enumerate}
\item $M'$ is a matching in $T'$.
\item $M'$ does not cover the vertex $v'$.
\item The weight of $M'$ is $\sum_{i=0}^n\theta^i2^{n-i}=(\theta^{n+1}-2^{n+1})/(\theta-2)$.
\end{enumerate}
\end{lemma}

\begin{proof}
By induction on $n$. For $n=0$, the claim is obvious from the description of {\sc{MakeTree}}. Assume that the claim holds for $n-1$, and consider the call {\sc{MakeTree}}$(n)$, which returns $(T',M',v')$. Then, by induction hypothesis, the two recursive calls must have returned $(T_1,M_1,v_1)$ and $(T_2,M_2,v_2)$ satisfying the conditions of the lemma. Suppose the algorithm replaced $(v_1,v_2)$ by $(v,v_1)$ in its matching. Since $v_1$ and $v_2$ were respectively left uncovered by $M_1$ and $M_2$, $M=M_1\cup M_2\cup\{(v_1,v_2)\}$ is a matching in $T$, and $M'$ does not cover $v'=v$. The case when the algorithm did not replace $(v_1,v_2)$ by $(v,v_1)$ is analogous. In either case, the additional edge in $M'$, apart from edges in $M_1$ and $M_2$ has weight $\theta^n$, and $M_1$, $M_2$ themselves have weight $\sum_{i=0}^{n-1}\theta^i2^{n-1-i}$, by induction hypothesis. Thus, the weight of $M'$ is $\theta^n+2\sum_{i=0}^{n-1}\theta^i2^{n-1-i}=\sum_{i=0}^n\theta^i2^{n-i}$.
\end{proof}

\begin{lemma}\label{lem_theta_alg}
When the call {\sc{MakeTree}}$(n)$ returns $(T,M,v)$, the algorithm's matching contains exactly one edge from $T$. This edge is incident on $v$ and has weight $\theta^n$.
\end{lemma}

\begin{proof}
By induction on $n$. Again, the claim is obvious for $n=0$. Assume that the claim holds for $n-1$, and consider the call {\sc{MakeTree}}$(n)$, which returns $(T',M',v')$. At the end of the two recursive calls which return $(T_1,M_1,v_1)$ and $(T_2,M_2,v_2)$. The algorithm will have exactly one edge $e_1$ from $T_1$ incident on $v_1$, and one edge $e_2$ from $T_2$ incident on $v_2$, by induction hypothesis. If the algorithm does not pick the next edge $(v_1,v_2)$, then the tree is discarded. If the algorithm picks that edge, then it must preempt $e_1$ and $e_2$. Thereafter, if the algorithm replaces $(v_1,v_2)$ by $(v,v_1)$ in its matching, then $v'=v$. Otherwise, if the algorithm keeps $(v_1,v_2)$, then $v'=v_2$. In either case, the algorithm is left with exactly one edge, which is incident on $v'$, and which has weight $\theta^n$.
\end{proof}

The adversary's strategy is given by Algorithm \ref{alg_adv}, where $n\geq1$ is a parameter.

\begin{algorithm}
\caption{{\sc{Adv}}$(n)$}
\begin{algorithmic}[1]
\WHILE{true}
\STATE $(T_1,M_1,v_1)$ $:=$ {\sc{MakeTree}}$(n-1)$
\STATE $(T_2,M_2,v_2)$ $:=$ {\sc{MakeTree}}$(n-1)$
\STATE Give the edge $(v_1,v_2)$ with weight $\theta^n$.
\IF{algorithm picks $(v_1,v_2)$}
\STATE Take a fresh vertex $v$, and give the edge $(v,v_1)$ with weight $\theta^n$.
\IF{algorithm replaces $(v_1,v_2)$ by $(v,v_1)$}
\STATE Take a fresh vertex $v'$ and give the edge $(v,v')$ with weight $\theta^n$.
\STATE $T:=T_1\cup T_2\cup\{(v_1,v_2),(v,v_1),(v,v')\}$.
\RETURN $M_1\cup M_2\cup\{(v_1,v_2),(v,v')\}$
\ELSE 
\STATE \COMMENT{algorithm still has $(v_1,v_2)$}
\STATE Take a fresh vertex $v'$ and give the edge $(v_2,v')$ with weight $\theta^n$.
\STATE $T:=T_1\cup T_2\cup\{(v_1,v_2),(v,v_1),(v_2,v')\}$.
\RETURN $M_1\cup M_2\cup\{(v,v_1),(v_2,v')\}$
\ENDIF
\ELSE 
\STATE \COMMENT{algorithm does not pick $(v_1,v_2)$}
\STATE Discard the constructed tree and retry. 
\ENDIF
\ENDWHILE
\end{algorithmic}
\label{alg_adv}
\end{algorithm}

\begin{lemma}\label{lem_discard}
When a tree $T$ is discarded in a call to {\sc{MakeTree}}$(n)$ or\\ {\sc{Adv}}$(n)$, $\alg(T)\leq\left(2+\frac{2}{\theta-2}\right)\cdot$ $\adv(T)$, where $\alg(T)$ and $\adv(T)$ are respectively the total weights of the edges of the algorithm's and the adversary's matchings, in $T$.
\end{lemma}

\begin{proof}
For $n\geq1$, consider the two calls to {\sc{MakeTree}}, which returned $(T_1,M_1,$ $v_1)$ and $(T_2,M_2,v_2)$ before the edge $(v_1,v_2)$ is revealed. By Lemma \ref{lem_theta_alg}, the algorithm had exactly one edge in each of $T_1$ and $T_2$, and this edge had weight $\theta^{n-1}$. The tree was discarded because the algorithm did not pick the edge $(v_1,v_2)$. Thus, $\alg(T)=2\theta^{n-1}$. On the other hand, the adversary picks the matching $M_1\cup M_2\cup\{(v_1,v_2)\}$ which, by Lemma \ref{lem_theta_adv}, has weight $\adv(T)=\theta^n+2\sum_{i=0}^{n-1}\theta^i2^{n-1-i}=\sum_{i=0}^n\theta^i2^{n-i}=(\theta^{n+1}-2^{n+1})/(\theta-2)$. Thus,
\begin{align*}
\frac{\adv(T)}{\alg(T)}&=\frac{\theta^{n+1}-2^{n+1}}{2\theta^{n-1}(\theta-2)}=\frac{\theta}{2}\times\frac{1-\left(\frac{2}{\theta}\right)^{n+1}}{1-\frac{2}{\theta}} \geq\frac{\theta}{2}\times\frac{1-\left(\frac{2}{\theta}\right)^2}{1-\frac{2}{\theta}}\\
  &=\frac{\theta}{2}\times\left(1+\frac{2}{\theta}\right)=\frac{\theta}{2}+1 \geq\left(2+\frac{2}{\theta-2}\right)
\end{align*}
The last inequality follows from the fact that $\theta\geq4$. Finally, note that when the discard happens in a call to {\sc{MakeTree}}$(0)$, $\alg(T)=0$ and $\adv(T)=1$.
\end{proof}

Now we are ready to prove Theorem \ref{thm_theta}.

\begin{proof}[Proof of Theorem \ref{thm_theta}]
For $n\geq1$, give the adversarial input by calling {\sc{Adv}}$(n)$. If the call does not terminate, then an unbounded number of trees are discarded, and by Lemma \ref{lem_discard}, a lower bound of $\left(2+\frac{2}{\theta-2}\right)$ is forced on each discarded tree. If the call terminates, then suppose $T$ is the final tree constructed. Let $(T_1,M_1,v_1)$ and $(T_2,M_2,v_2)$ be returned by the two calls to {\sc{MakeTree}}$(n-1)$. By the description of {\sc{Adv}} and Lemma \ref{lem_theta_alg}, it is clear that the algorithm holds only one edge of $T$ in the end, and this edge has weight $\theta^n=\alg(T)$. On the other hand, the adversary's matching contains $M_1$ and $M_2$, and two edges of weight $\theta^n$, where by Lemma \ref{lem_theta_adv}, the weight of $M_1$ and $M_2$ is $(\theta^n-2^n)/(\theta-2)$ each. Thus, $\adv(T)=2\theta^n+2(\theta^n-2^n)/(\theta-2)$. Therefore,
\[\frac{\adv(T)}{\opt(T)}=2+\frac{2(\theta^n-2^n)}{\theta^n(\theta-2)}=2+\frac{2\left(1-\left(\frac{2}{\theta}\right)^n\right)}{\theta-2}\]
This approaches $\left(2+\frac{2}{\theta-2}\right)$ as $n\rightarrow\infty$. Furthermore, this lower bound is also forced on the trees discarded during the execution of {\sc{Adv}}$(n)$. Thus, the algorithm can not have a competitive ratio less than $\left(2+\frac{2}{\theta-2}\right)$.
\end{proof}

\section{Proof of Theorem~\ref{th1}}\label{opt_paths}

Theorem~\ref{th1} can be proved using the following lemma.
\begin{lemma}
For any (maximal) path $P$ of length $n>0$,
\begin{itemize}
 \item if $n$ is even then $\mathbb{E}[|M\cap P|] \geq (3/4)(n/2)+1/4 = p_0(n)$ (say).
 \item if $n$ is odd then $\mathbb{E}[|M\cap P|]\geq (3/4)(n/2)+ 3/8 = p_1(n)$ (say).
 \end{itemize}
\end{lemma}
\begin{proof}For $n=1$ and $n=2$, $ \mathbb{E}[|M\cap P|]=1$, and for $n=3$, $ \mathbb{E}[|M\cap P|]=\frac{3}{2}$. Thus the lemma holds when $n\leq 3$. We will induct on the number of edges in the input. (Case $n=1$ covers the base case.) Suppose the lemma is true before the arrival of the new edge $e$. We prove that the lemma holds even after $e$ has been processed. We may assume that the length $n$ of the new path $P$ resulting from addition of $e$ is at least $4$.
\begin{enumerate}
\begin{comment}
 \item If a path of length $5$ is formed by joining two paths of length $2$ each, by a new edge: With probability $3/4$ atleast one of the joining edges of two paths will be present and with probability $1/4$, none of them will be present.
 \begin{align*}
  E\left[path_{5}^{odd}\right]&= \frac{3}{4}\left(E\left[path_{2}^{even}\right]+E\left[path_{2}^{even}\right]\right) + \frac{1}{4}\left(E\left[path_{e}^{even}\right]+E\left[path_{2}^{even}\right] + 1\right)\\
                    &= 2+\frac{1}{4}\\
                    &= \frac{3}{4}\left(\frac{5+1}{2}\right)
 \end{align*}
\end{comment}
 \item If $n$ is even and $L_2=0$, (therefore $L_1$ is odd, and $L_1 \geq 3$), $\Pr[e_1\notin M]=\frac{1}{2}$. Therefore, $e$ is added to $M$ with probability $\frac{1}{2}$.
 \begin{align*}
  \mathbb{E}\left[|M\cap P|\right] \geq p_1(n-1)+\frac{1}{2} \geq p_0(n)
 \end{align*}

 \item If $n$ is even, $L_1=n-2$ and $L_2=1$.
 \begin{align*}
  \mathbb{E}\left[|M\cap P|\right]&\geq p_0(n-2)+1\\
                    &= \frac{3}{4}\left(\frac{n-2}{2}\right)+\frac{1}{4}+1\\
                    &\geq p_0(n)
 \end{align*}
\begin{comment}
 \item If path of length $n$, such that $n$ is even, is formed by joining a path of length $n-3$ and a path of length $2$ by a new edge: With probability $3/4$ atleast one of the joining edges of two paths will be present and with probability $1/4$, none of them will be present.
 \begin{align*}
  E\left[path_{n}^{even}\right]&= \frac{3}{4}\left(E\left[path_{n-3}^{odd}\right]+E\left[path_{2}^{even}\right]\right) + \frac{1}{4}\left(E\left[path_{n-3}^{odd}\right]+E\left[path_{2}^{even}\right] + 1\right)\\
                    &\geq \frac{3}{4}\left(\frac{n-2}{2}\right)+1+\frac{1}{4}\\
                    &= \frac{3}{4}\left(\frac{n}{2}\right)+\frac{1}{2}
 \end{align*}
\end{comment}
 \item If $n$ is even, and $L_2>1$, where $n=L_1+L_2+1$, $L_1$ is even, and $L_2$ is odd. $\Pr[e_1\notin M, e_2\notin M]=\frac{1}{4}$
 \begin{align*}
  \mathbb{E}\left[|M\cap P|\right]&\geq p_0(L_1)+p_1(L_2)+\frac{1}{4}\\
                    &= \frac{3}{4}\left(\frac{L_1+L_2+1}{2}\right)+\frac{1}{4}+\frac{3}{8}-\frac{3}{8}+\frac{1}{4}\\
                    &\geq p_0(n)
 \end{align*}

 \item If $n$ is odd, and $L_2=0$, (therefore $L_1$ is even, and $L_1 \geq 3$), $\Pr[e_1\notin M]=\frac{1}{2}$. Therefore, $e$ is added to $M$ with probability $\frac{1}{2}$.
 \begin{align*}
  \mathbb{E}\left[|M\cap P|\right] \geq p_0(n-1)+\frac{1}{2} = p_1(n)
 \end{align*}

 \item If $n$ is odd, $L_1=n-2$ and $L_2=1$.
 \begin{align*}
  \mathbb{E}\left[|M\cap P|\right]&\geq p_1(n-2)+1\\
                    &= \frac{3}{4}\left(\frac{n-2}{2}\right)+\frac{3}{8}+1\\
                    &\geq p_1(n)
 \end{align*}
\begin{comment}
 \item If path of length $n$, such that $n$ is odd, is formed by joining a path of length $n-3$ and a path of length $2$ by a new edge: With probability $3/4$ atleast one of the joining edges of two paths will be present and with probability $1/4$, none of them will be present.
 \begin{align*}
  E\left[path_{n}^{odd}\right]&= \frac{3}{4}\left(E\left[path_{n-3}^{even}\right]+E\left[path_{2}^{even}\right]\right) + \frac{1}{4}\left(E\left[path_{n-3}^{even}\right]+E\left[path_{2}^{even}\right] + 1\right)\\
                    &\geq \frac{3}{4}\left(\frac{n-3}{2}\right)+\frac{1}{2}+1+\frac{1}{4}\\
                    &\geq \frac{3}{4}\left(\frac{n+1}{2}\right)
 \end{align*}
\end{comment}
 \item If $n$ is odd, and $L_2>1$, where $n=L_1+L_2+1$, $L_1$ is even, and $L_2$ is even. $\Pr[e_1\notin M, e_2\notin M]=\frac{1}{4}$
 \begin{align*}
  \mathbb{E}\left[|M\cap P|\right]&\geq p_0(L_1)+p_0(L_2)+\frac{1}{4}\\
                    &= \frac{3}{4}\left(\frac{L_1+L_2+1}{2}\right)+\frac{1}{4}+\frac{1}{4}-\frac{3}{8}+\frac{1}{4}\\
                    &= p_1(n)
 \end{align*}

 \item If $n$ is odd, and $L_2>1$, where $n=L_1+L_2+1$, $L_1$ is odd, and $L_2$ is odd. $\Pr[e_1\notin M, e_2\notin M]=\frac{1}{4}$
 \begin{align*}
  \mathbb{E}\left[|M\cap P|\right]&\geq p_1(L_1)+p_1(L_2)+\frac{1}{4}\\
                    &= \frac{3}{4}\left(\frac{L_1+L_2+1}{2}\right)+\frac{3}{8}+\frac{3}{8}-\frac{3}{8}+\frac{1}{4}\\
                    &\geq p_1(n)
 \end{align*}

\end{enumerate}
This completes the induction and hence implies a $\frac{4}{3}$-competitive ratio for this algorithm.
\end{proof}

\end{document}